\documentclass[11pt,letterpaper]{article}

\usepackage{url}

\usepackage[letterpaper,margin=1.00in]{geometry}
\usepackage{times}
\usepackage{amsmath, amssymb, amsthm}
\usepackage{cite}
\usepackage{appendix}
\usepackage{graphicx}
\usepackage{epstopdf}
\usepackage{color}
\usepackage{algorithm}
\usepackage[noend]{algpseudocode}
\usepackage{multirow}

\usepackage{xspace}
\usepackage{mdwlist}
\usepackage[framemethod=tikz]{mdframed}
\usepackage{caption}
\usepackage[colorinlistoftodos]{todonotes}

\definecolor{dgreen}{rgb}{0, 0.5,0}
\usepackage{hyperref}
\hypersetup{
    unicode=false,          
    colorlinks=true,        
    linkcolor=red,          
    citecolor=dgreen,        
    filecolor=magenta,      
    urlcolor=cyan           
}
\usepackage[capitalize]{cleveref}

\algnewcommand\algorithmicswitch{\textbf{switch}}
\algnewcommand\algorithmiccase{\textbf{case}}

\algdef{SE}[SWITCH]{Switch}{EndSwitch}[1]{\algorithmicswitch\ #1\ \algorithmicdo}{\algorithmicend\ \algorithmicswitch}%
\algdef{SE}[CASE]{Case}{EndCase}[1]{\algorithmiccase\ #1}{\algorithmicend\ \algorithmiccase}%
\algtext*{EndSwitch}%
\algtext*{EndCase}%

\newcommand{\eps}{\varepsilon}

\newcommand{\set}[1]{\left\{#1\right\}}
\newcommand{\C}{{\lambda}}

\newcommand{\local}{\ensuremath{\mathsf{LOCAL}}}
\newcommand{\congest}{\ensuremath{\mathsf{CONGEST}}}

\renewcommand{\paragraph}[1]{\vspace{0.15cm}\noindent {\bf #1}}
\newtheorem{theorem}{Theorem}[section]
\newtheorem{lemma}[theorem]{Lemma}
\newtheorem{claim}[theorem]{Claim}

\newtheorem{corollary}[theorem]{Corollary}

\newtheorem{definition}[theorem]{Definition}

\newcommand{\FullOrShort}{full}

\ifthenelse{\equal{\FullOrShort}{full}}{
	
	  \newcommand{\fullOnly}[1]{#1}
	  \newcommand{\shortOnly}[1]{}
		\newcommand{\IncludePictures}[1]{#1}
		\newcommand{\algorithmsize}{\footnotesize}
	}{

  \usepackage{times}
    
		\newcommand{\fullOnly}[1]{}
	  \newcommand{\shortOnly}[1]{#1}
		\newcommand{\IncludePictures}[1]{}
		\newcommand{\algorithmsize}{\footnotesize}
	}

\begin{document}

\title{Distributed Minimum Cut Approximation\footnote{A preliminary
    version of this paper appeared in \cite{disc-version}.}}

\author{
  Mohsen Ghaffari\\
  \small MIT\\
  \small Cambridge, MA, USA\\
  \texttt{\small ghaffari@mit.edu}
  \and
  Fabian Kuhn\\
  \small University of Freiburg\\
  \small Freiburg, Germany\\
  \texttt{\small kuhn@cs.uni-freiburg.de}
}

\date{}




\maketitle

\begin{abstract}
  We study the problem of computing approximate minimum edge cuts by
  distributed algorithms. We use a standard synchronous message passing
  model where in each round, $O(\log n)$ bits can be transmitted over
  each edge (a.k.a.\ the $\congest$ model). We present a distributed algorithm that, for any weighted graph and any $\epsilon \in (0, 1)$, with high probability finds a cut of size at most $O(\epsilon^{-1}\lambda)$ in $O(D) + \tilde{O}(n^{1/2 + \epsilon})$ rounds, where $\lambda$ is the size of the minimum cut. This algorithm is based on a simple approach for analyzing random edge sampling, which we call the \emph{random layering technique}. In addition, we also present another distributed algorithm, which is based on a centralized algorithm due to Matula [SODA '93], that with high probability computes a cut of size at most $(2+\epsilon)\lambda$ in $\tilde{O}((D+\sqrt{n})/\epsilon^5)$ rounds for any $\epsilon>0$.
	
	The time complexities of our algorithms
  almost match the $\tilde{\Omega}(D + \sqrt{n})$ lower bound of Das
  Sarma et al.~[STOC '11], thus leading to an answer to an open
  question raised by Elkin~[SIGACT-News '04] and Das Sarma et
  al.~[STOC '11].
	
  To complement our upper bound results, we also strengthen the
  $\tilde{\Omega}(D + \sqrt{n})$ lower bound of Das Sarma et al.\ by
  extending it to unweighted graphs. We show that the same lower bound
  also holds for unweighted multigraphs (or equivalently for weighted
  graphs in which $O(w\log n)$ bits can be transmitted in each round
  over an edge of weight $w$). These results even hold if the diameter
  is $D=O(\log n)$. For unweighted simple graphs, we show that even
  for networks of diameter $\tilde{O}\big(\frac{1}{\lambda}\cdot
  \sqrt{\frac{n}{\alpha\lambda}}\big)$ finding an $\alpha$-approximate
  minimum cut in networks of edge connectivity $\lambda$ or computing
  an $\alpha$-approximation of the edge connectivity requires time at
  least $\tilde{\Omega}\big(D + \sqrt{\frac{n}{\alpha\lambda}}\big)$.
\end{abstract}

\section{Introduction}
\label{sec:intro}

Finding minimum cuts or approximately minimum cuts are classical and
fundamental algorithmic graph problems with many important
applications. In particular, minimum edge cuts and their size (i.e.,
the edge connectivity) are relevant in the context of networks, where
edge weights might represent link capacities and therefore edge
connectivity can be interpreted as the throughput capacity of the
network. Decomposing a network using small cuts helps designing
efficient communication strategies and finding communication
bottlenecks (see, e.g., \cite{PQ82, KargerStein93}). Both the exact
and approximate variants of the minimum cut problem have received
extensive attention in the domain of centralized algorithms (cf.\
Section \ref{subsec:centralized} for a brief review of the results in
the centralized setting). This line of research has led to (almost)
optimal centralized algorithms with running times $\tilde{O}(m+n)$
\cite{Karger-LinearTimeCut} for the exact version and
$O(m+n)$\cite{Matula93} for constant-factor approximations, where $n$
and $m$ are the numbers of nodes and edges, respectively.

As indicated by Elkin\cite{elkin04} and Das Sarma et
al. \cite{dassarma12}, the problem has remained essentially open in
the distributed setting. In the $\local$ model \cite{peleg00} where in
each round, a message of unbounded size can be sent over each edge,
the problem has a trivial time complexity of $\Theta(D)$ rounds, where
$D$ is the (unweighted) diameter of the network. The problem is
therefore more interesting and also practically more relevant in
models where messages are of some bounded size $B$. The standard model
incorporating this restriction is the $\congest$ model \cite{peleg00},
a synchronous message passing model where in each time unit, $B$ bits
can be sent over every link (in each direction). It is often assumed
that $B=\Theta(\log n)$. The only known non-trivial result is an
elegant lower bound by Das Sarma et al.~\cite{dassarma12} showing
that any $\alpha$-approximation of the minimum cut in weighted graphs
requires at least $\Omega(D + \sqrt{n/(B\log n)})$ rounds.

\paragraph{Our Contribution:}
We present two distributed minimum-cut
approximation algorithms for undirected weighted graphs, with complexities almost matching the lower bound
of \cite{dassarma12}. We also extend the lower
bound of \cite{dassarma12} to unweighted graphs and multigraphs.

Our first algorithm, presented in \Cref{sec:layeringCut}, with high
probability\footnote{We use the phrase \emph{with
    high probability} (w.h.p.) to indicate
  probability greater than $1-\frac{1}{n}$.} finds a cut of size at most
$O(\eps^{-1}\C)$, for any $\epsilon \in (0, 1)$ and where $\C$ is the
edge connectivity, i.e., the size of the minimum cut in the
network. The time complexity of this algorithm is $O(D)+
O(n^{1/2+\epsilon} \log^3 n \log \log n \log^* n)$. The algorithm is
based on a simple and novel approach for analyzing random edge
sampling, a tool that has proven extremely successful also for
studying the minimum cut problem in the centralized setting (see,
e.g., \cite{KargerStein93}). Our analysis is based on \emph{random
  layering}, and we believe that the approach might also be useful for
studying other connectivity-related questions. Assume that each edge
$e\in E$ of an unweighted multigraph $G=(V, E)$ is independently
sampled and added to a subset $E' \subset E$ with probability
$p$. For $p \leq \frac{1}{\C}$, the graph $G'=(V, E')$ induced by the sampled edges is
disconnected with at least a constant probability (just consider one min-cut). In Section
\ref{sec:sampling}, we use random layering to show that
if $p = \Omega(\frac{\log n}{\C})$, the sampled graph $G'$ is
connected w.h.p. This bound is optimal and was known previously, with
two elegant proofs:~\cite{LP71} and~\cite{KargerSTOC94}. Our proof is simple and self-contained and it
serves as a basis for our algorithm in Section \ref{sec:layeringCut}.

The second algorithm, presented in \Cref{sec:Matula}, finds a cut with
size at most $(2+\eps)\C$, for any constant $\eps >0$, in time $O((D+
\sqrt{n}\log^* n) \log^2 n \log \log n \cdot \frac{1}{\eps^5})$. This
algorithm combines the general approach of Matula's centralized
$(2+\eps)$-approximation algorithm \cite{Matula93} with Thurimella's
algorithm for sparse edge-connectivity certificates\cite{Thurimella97}
and with the famous random edge sparsification technique of Karger
(see e.g., \cite{KargerSTOC94}).

To complement our upper bounds, we also extend the lower bound of Das
Sarma et al.~\cite{dassarma12} to unweighted graphs and
multigraphs. When the minimum cut problem (or more generally problems
related to small edge cuts and edge connectivity) are in a distributed
context, often the weights of the edges correspond to their
capacities. It therefore seems reasonable to assume that over a link
of twice the capacity, we can also transmit twice the amount of data
in a single time unit. Consequently, it makes sense to assume that
over an edge of weight (or capacity) $w\geq 1$, $O(w\log n)$ bits can
be transmitted per round (or equivalently that such a link corresponds
to $w$ parallel links of unit capacity). The lower bound of
\cite{dassarma12} critically depends on having links with (very) large
weight over which in each round only $O(\log n)$ bits can be
transmitted. We generalize the approach of \cite{dassarma12} and
obtain the same lower bound result as in \cite{dassarma12} for the
weaker setting where edge weights correspond to edge capacities (i.e.,
the setting that can be modeled using unweighted
multigraphs). Formally, we show that if $Bw$ bits can be transmitted
over every edge of weight $w\geq 1$, for every $\alpha\geq 1$ and
sufficiently large $\lambda$, there are $\lambda$-edge-connected
networks with diameter $O(\log n)$ on which computing an
$\alpha$-approximate minimum cut requires time at least
$\Omega\big(\sqrt{n/(B\log n)}\big)$. Further, for unweighted simple
graphs with edge connectivity $\lambda$, we show that even for
diameter $D=\big(\frac{1}{\lambda}\cdot\sqrt{n/(\alpha\lambda B\log
  n)}\big)$ finding an $\alpha$-approximate minimum cut or
approximating the edge connectivity by a factor of $\alpha$ requires
at least time $\Omega\big(\sqrt{n/(\alpha\lambda B\log n)}\big)$.

In addition, 
our technique yields a structural result about
$\C$-edge-connected graphs with small diameter. We show that for every
$\C> 1$, there are $\C$-edge-connected graphs $G$ with diameter
$O(\log n)$ such that for any partition of the edges of $G$ into
spanning subgraphs, all but $O(\log n)$ of the spanning subgraphs have
diameter $\Omega(n)$ (in the case of unweighted multigraphs)
or $\Omega(n/\C)$ (in the case of unweighted simple graphs). As a
corollary, we also get that when sampling each edge of such a graph
with probability $p\leq \gamma/\log n$ for a sufficiently small
constant $\gamma>0$, with at least a positive constant probability, the subgraph
induced by the sampled edges has diameter $\Omega(n)$ (in the
case of unweighted multigraphs) and $\Omega(n/\C)$ (in the case of
unweighted simple graphs). The details
  about these results are deferred to \fullOnly{\Cref{sec:disseminationlower}}\shortOnly{the full version \cite{Cut-FullVersion}}. 

\subsection{Related Work in the Centralized Setting}\label{subsec:centralized}
Starting in the 1950s \cite{FF56, EFS56}, the
traditional approach to the minimum cut problem was to use max-flow
algorithms (cf.\ \cite{FordFulkerson} and \cite[Section 1.3]{KargerStein93}). In the
1990s, three new approaches were introduced which go away from the
flow-based method and provide faster algorithms: The first method, presented by Gabow\cite{Gabow91}, is
based on a matroid characterization of the min-cut and it
finds a min-cut in $O(m+\C^2 n\log{\frac{n}{m}})$ steps, for any
unweighted (but possibly directed) graph with edge connectivity $\C$.
The second approach applies to (possibly) weighted but undirected
graphs and is based on repeatedly identifying and contracting edges
outside a min-cut until a min-cut becomes apparent (e.g.,
\cite{NI92, Karger93-Contraction, KargerStein93}). The beautiful
\emph{random contraction algorithm} (RCA) of
Karger~\cite{Karger93-Contraction} falls into this category. In the
basic version of RCA, the following procedure is repeated $O(n^2 \log
n)$ times: contract uniform random edges one by one until only two
nodes remain. The edges between these two nodes correspond to a cut in
the original graph, which is a min-cut with probability at least
$1/O({n^2})$. Karger and Stein \cite{KargerStein93} also present a
more efficient implementation of the same basic idea, leading to total
running time of $O(n^2 \log^3 n)$. The third method, which again
applies to (possibly) weighted but undirected graphs, is due to
Karger\cite{KargerSTOC96} and is based on a ``semiduality" between
minimum cuts and maximum spanning tree packings. This third method
leads to the best known centralized minimum-cut
algorithm\cite{Karger-LinearTimeCut} with running time $O(m \log^3
n)$.

For the approximation version of the problem (in undirected graphs),
the main known results are as follows. Matula~\cite{Matula93} presents
an algorithm that finds a $(2+\eps)$-minimum cut for any constant
$\eps>0$ in time $O((m+n)/\eps)$. This algorithm is based on a graph
search procedure called \emph{maximum adjacency
  search}. Based on a modified version of the random contraction
algorithm, Karger\cite{KargerSODA94} presents an algorithm that finds
a $(1+\eps)$-minimum cut in time $O(m+n\log^3 n/\eps^4)$.


\section{Preliminaries}
\label{sec:prelim}

\paragraph{Notations and Definitions:}
We usually work with an undirected weighted graph $G=(V, E, w)$, where
$V$ is a set of $n$ vertices, $E$ is a set of (undirected) edges
$e=\{v,u\}$ for $u,v \in V$, and $w:E \rightarrow \mathbb{R}^{+}$ is a
mapping from edges $E$ to positive real numbers. For each edge
$e\in E$, $w(e)$ denotes the weight of edge $e$. In the special case of unweighted graphs, we simply assume $w(e)=1$ for
each edge $e \in E$.

For a given non-empty proper subset $C \subset V$, we define the cut
$(C, V \setminus C)$ as the set of edges in $E$ with exactly one endpoint
in set $C$. The size of this cut, denoted by $w(C)$ is the sum of the weights of the edges in set
$(C, V \setminus C)$. The edge-connectivity $\C(G)$ of the graph is
defined as the minimum size of $w(C)$ as $C$ ranges over all nonempty
proper subsets of $V$. A cut $(C, V\setminus C)$ is called
\emph{$\alpha$-minimum}, for an $\alpha\geq1$, if $w(C) \leq \alpha
\C(G)$. When clear from the context, we sometimes use $\C$ to
refer to $\lambda(G)$.

\paragraph{Communicaton Model and Problem Statements:}
We use a standard \emph{message passing model} (a.k.a.\ the $\congest$
model\cite{peleg00}), where the execution proceeds in synchronous
rounds and in each round, each node can send a message of size $B$ bits to
each of its neighbors. A typically standard case is $B=\Theta(\log n)$.

For upper bounds, for simplicity we assume that
$B=\Theta(\log n)$\footnote{Note that by choosing $B=b\log n$ for
  some $b\geq 1$, in all our upper bounds, the term that does not
  depend on $D$ could be improved by a factor $\sqrt{b}$.}. For upper bounds, we further assume that $B$ is large enough so that a
constant number of node identifiers and edge weights can be packed
into a single message. For $B=\Theta(\log n)$, this implies that each
edge weight $w(e)$ is at most (and at least) polynomial in
$n$. W.l.o.g., we further assume that edge weights are normalized
and each edge weight is an integer in range $\{1, \dots,
n^{\Theta(1)}\}$. Thus, we can also view a weighted graph as a
multi-graph in which all edges have unit weight and multiplicity at
most $n^{\Theta(1)}$ (but still only $O(\log n)$ bits can be
transmitted over all these parallel edges together).

For lower bounds, we assume a weaker model where $B\cdot w(e)$ bits
can be sent in each round over each edge $e$. To ensure that at least
$B$ bits can be transmitted over each edge, we assume that the weights
are scaled such that $w(e)\geq 1$ for all edges. For integer weights,
this is equivalent to assuming that the network graph is an unweighted
multigraph where each edge $e$ corresponds to $w(e)$ parallel
unit-weight edges.

In the problem of computing an
\emph{$\alpha$-approximation of the minimum cut}, the goal is to find
a cut $(C^*, V\setminus C^*)$ that is $\alpha$-minimum. To indicate
this cut in the distributed setting, each node $v$ should know whether $v \in C^*$. In the
problem of \emph{$\alpha$-approximation of the edge-connectivity}, all nodes must output an estimate $\tilde{\C}$ of $\C$ such that
$\tilde{\C}\in [\C, \C\alpha]$. In randomized algorithms for these problems, time complexities are fixed deterministically and the correctness guarantees are required to hold with high probability.


\subsection{Black-Box Algorithms}\label{subsec:Thurimella}

In this paper, we make frequent use of a \emph{connected component
  identification} algorithm due to Thurimella \cite{Thurimella97},
which itself builds on the minimum spanning tree algorithm of Kutten
and Peleg\cite{KuttenPeleg95}. Given a graph $G(V,E)$ and a subgraph
$H=(V, E')$ such that $E' \subseteq E$, Thurimella's algorithm
identifies the connected components of $H$ by assigning a label
$\ell(v)$ to each node $v \in V$ such that two nodes get the same
label iff they are in the same connected component of $H$. The time
complexity of the algorithm is $O(D+\sqrt{n} \log^{*}{n})$ rounds,
where $D$ is the (unweighted) diameter of $G$. Moreover, it is easy to
see that the algorithm can be made to produce labels $\ell(v)$ such
that $\ell(v)$ is equal to the smallest (or the largest) id in the
connected component of $H$ that contains $v$.
Furthermore, the connected component identification algorithm can also
be used to test whether the graph $H$ is connected (assuming that $G$
is connected). $H$ is not connected if and only if there is an edge
$\set{u,v}\in E$ such that $\ell(u)\neq\ell(v)$. If some node $u$
detects that for some neighbor $v$ (in $G$), $\ell(u)\neq\ell(v)$, $u$
broadcasts \emph{not connected}. Connectivity of $H$ can therefore be
tested in $D$ additional rounds. We refer to this as Thurimella's
\emph{connectivity-tester} algorithm. Finally, we remark that the same
algorithms can also be used to solve $k$ independent instances of the
connected component identification problem or $k$ independent
instances of the connectivity-testing problem in $O(D+k\sqrt{n}
\log^{*}{n})$ rounds. This is achieved by pipelining the messages of
the broadcast parts of different instances.

\section{Edge Sampling and The Random Layering Technique} 
\label{sec:sampling}

Here, we study the process of random edge-sampling and present a simple technique, which we call \emph{random
  layering}, for analyzing the connectivity of the graph obtained
through sampling. This technique also forms the basis
of our min-cut approximation algorithm presented in the next section.

\paragraph{Edge Sampling} Consider an arbitrary unweighted multigraph
$G=(V,E)$. Given a probability $p \in [0,1]$, we define an \emph{edge
  sampling experiment} as follows: choose subset $S \subseteq E$ by
including each edge $e \in E$ in set $S$ independently with
probability $p$. We call the graph $G'=(V, S)$ \emph{the sampled
  subgraph}.

We use the \emph{random layering technique} to answer the following \emph{network reliability} question: ``How
large should $p$ be, as a function of minimum-cut size $\C$,
so that the sampled graph is connected w.h.p.?"\footnote{A rephrased version is, how large should the edge-connectivity $\C$ of a network be such
  that it remains connected w.h.p. if each edge fails with probability
  $1-p$.} 
Considering just one cut of size $\C$ we see that if $p \leq \frac{1}{\C}$, then the probability
that the sampled subgraph is connected is at most $\frac{1}{e}$. We
show that $p \geq\frac{20\log n}{\C}$ suffices so that the sampled subgraph is connected w.h.p. Note that this is non-trivial as
a graph has exponential many cuts. It is easy to see that this bound is
asymptotically optimal~\cite{LP71}.

\begin{theorem}
\label{thm:UpperTh} Consider an arbitrary unweighted multigraph $G=(V, E)$ with edge connectivity $\C$ and choose subset $S \subseteq E$ by including each edge $e \in E$ in set $S$ independently with probability $p$. If $p \geq \frac{20\log n}{\C}$, then the sampled subgraph $G'=(V, S)$ is connected with probability at least $1-\frac{1}{n}$.
\end{theorem}

We remark that this result was known prior to this paper, via two
different proofs by Lomonosov and Polesskii \cite{LP71} and
Karger\cite{KargerSTOC94}. The Lomonosov-Polesskii proof \cite{LP71} uses
an interesting coupling argument and shows that among the graphs of a
given edge-connectivity $\C$, a cycle of length $n$ with edges of
multiplicity $\C/2$ has the smallest probability of remaining
connected under random sampling. Karger's proof\cite{KargerSTOC94} uses
the powerful fact that the number of $\alpha$-minimum cuts is at most
$O(n^{2\alpha})$ and then uses basic probability concentration
arguments (Chernoff and union bounds) to show that, w.h.p., each cut
has at least one sampled edge. There are many known proofs for the
$O(n^{2\alpha})$ upper bound on the number of $\alpha$-minimum cuts
(see \cite{Karger-LinearTimeCut}); an elegant argument follows from
Karger's \emph{random contraction
  algorithm}\cite{Karger93-Contraction}.

Our proof of \Cref{thm:UpperTh} is simple and
self-contained, and it is the only one of the three approaches that
extends to the case of random vertex failures\footnote{There, the question is, how large the vertex sampling
  probability $p$ has to be chosen, as a function of vertex connectivity $k$, so that the
  vertex-sampled graph is connected, w.h.p. The extension to the
  vertex version requires important modifications and leads
  to $p = \Omega(\frac{\log n}{\sqrt{k}})$ being a sufficient
  condition. Refer to \cite[Section 3]{CGK13} for details.}
\cite[Theorem 1.5]{CGK13}.

\begin{proof}[Proof of \Cref{thm:UpperTh}]
  Let $L = 20\log n$. For each edge $e\in E$, we independently choose
  a uniform random \emph{layer number} from the set $\{1, 2, \dots,
  L\}$. Intuitively, we add the sampled edges layer by layer and
  show that with the addition of the sampled edges of each layer, the number of
  connected components goes down by at least a constant factor, with
  at least a constant probability, and independently of the previous
  layers. After $L=\Theta(\log n)$ layers, connectivity is achieved w.h.p.

  We start by presenting some notations. For each $i \in \{1, \dots, L\}$, let $S_i$ be the set
  of sampled edges with layer number $i$ and let $S_{i-} =
  \bigcup_{j=1}^{i} S_j$, i.e., the set of all sampled edges in layers $\{1, \dots, i\}$. Let $G_i = (V, S_{i-})$ and let $M_i$ be the
  number of connected components of graph $G_i$. We show that $M_{L}=1$, w.h.p. 
	
	For any $i \in [1, L-1]$, since $S_{i-} \subseteq S_{(i+1)-}$, we have $M_{i+1} \leq
  M_{i}$. Consider the indicator variable ${X}_i$ such that ${X}_i =1$
  iff $M_{i+1} \leq 0.87 M_{i}$ or $M_i=1$. We show the following
  claim, after which, applying a Chernoff bound completes the proof.
  \begin{claim}\label{clm}
    For all $i \in [1, L-1]$ and $T \subseteq E$, we have  $\Pr[X_i=1| S_{i-}=T] \geq
    1/2$.
  \end{claim}

	To prove this claim, we use \emph{the principle of deferred
  decisions}\cite{StableMarriage} to view the two random processes
  of sampling edges and layering them. More
  specifically, we consider the following process: first, each edge is
  sampled and given layer number $1$ with probability $p/L$. Then,
  each remaining edge is sampled and given layer number $2$ with
  probability $\frac{p/L}{1-p/L} \geq p/L$. Similarly, after
  determining the sampled edges of layers $1$ to $i$, each remaining
  edge is sampled and given layer number $i+1$ with probability
  $\frac{p/L}{1-(i\,p)/L} \geq p/L$. After doing
  this for $L$ layers, any remaining edge is considered not sampled
  and it receives a random layer number from $\{1, 2, \dots, L\}$. It
  is easy to see that in this process, each edge is independently
  sampled with probability exactly $p$ and each edge $e$ gets a
  uniform random layer number from $\{1, 2, \dots, L\}$, chosen
  independently of the other edges and also independently of whether
  $e$ is sampled or not.
	  
	{
  \begin{figure}[t]
    \centering
    \includegraphics[width=0.45\textwidth]{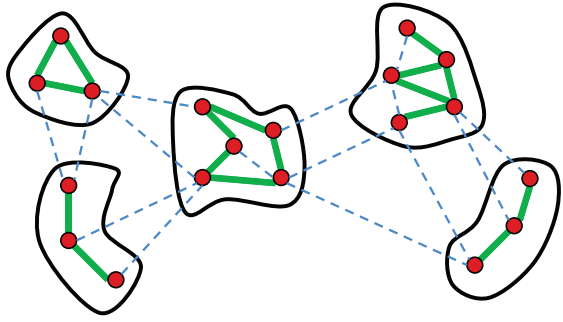}
    {\caption{\small Graph $G_i$ and its connected components. The green solid
        links represent edges in $S_{i-}$ and the blue dashed links represent $E
        \setminus S_{i-}$.}}
    \label{fig:Fig1}
		\vspace{-0.3cm}
  \end{figure}
	}

  Fix a layer $i \in [1, \dots, L-1]$ and a subset $T \subseteq
  E$. Let $S_{i-}=T$ and consider graph $G_i=(V,
  S_{i-})$. \Cref{fig:Fig1} presents an example
  graph $G_i$ and its connected components.
  If $M_i=1$ meaning that $G_i$ is connected, then
  $X_{i}=1$. Otherwise, suppose that $M_i \geq 2$. For each component
  $\mathcal{C}$ of $G_i$, call the component \emph{bad} if
  $(\mathcal{C}, V\setminus \mathcal{C}) \cap S_{i+1} =
  \emptyset$. That is, $\mathcal{C}$ is bad if after adding the
  sampled edges of layer $i+1$, $\mathcal{C}$ does not get connected
  to any other component. We show that $\Pr[\mathcal{C}\textit{ is
    bad}] \leq \frac{1}{e}$. 
		
		Since $G$ is $\C$-edge connected, we have
  $w(C) \geq \C$. Moreover, none of the edges in $(\mathcal{C},
  V\setminus \mathcal{C})$ is in $S_{i-}$. Thus, using the principle
  of deferred decisions as described, each of the edges of the cut
  $(\mathcal{C}, V\setminus \mathcal{C})$ has probability
  $\frac{p/L}{1-(i\, p)/L} \geq p/L$ to be sampled and given layer
  number $i+1$, i.e., to be in $S_{i+1}$. Since $p
  \geq \frac{20\log n}{\C}$, the probability that none of the edges
  $(\mathcal{C}, V\setminus \mathcal{C})$ is in set $S_{i+1}$ is at
  most $(1-p/L)^{\C} \leq 1/e$. Thus, $\Pr[\mathcal{C}\textit{ is
    bad}] \leq 1/e$.
	\shortOnly{Having this, since each component that is not
  bad gets connected to at least one other component (when we look at
  graph $G_{i+1}$), a simple application of Markov's inequality proves the claim, and after that, a Chernoff bound completes the proof. See \cite{Cut-FullVersion} for details.}
	\fullOnly{
	
	Now let $Z_i$ be the number of bad components of $G_i$. Since each
  component is bad with probability at most $1/e$, we have
  $\mathbb{E}[Z_i]\leq M_i/e$. Using Markov's inequality, we get $\Pr[Z_i \geq 2M_i/e] \leq 1/2$. Since each component that is not
  bad gets connected to at least one other component (when we look at
  graph $G_{i+1}$), we have $M_{i+1} \leq Z_i + \frac{(M_i - Z_i)}{2}
  = \frac{M_i + Z_i}{2}$. Therefore, with probability at least $1/2$,
  we have $M_{i+1} \leq \frac{1+2/e}{2} M_{i} < 0.87 M_{i}$. This means that
  $\Pr[X_i=1] \geq 1/2$, which concludes the proof of the claim.

  Now using the claim, we get that $\mathbb{E}[\sum_{i=1}^{L-1} X_i]
  \geq 10\log n$. A Chernoff bound then shows that
  $\Pr[\sum_{i=1}^{L-1} X_i \geq 5\log n] \geq 1-\frac{1}{n}$. This
  means that w.h.p, $M_L \leq \frac{n}{2^{\log n}}
  =1$. That is, w.h.p, $G_L=(V,S)=(V,S_{L-})=G'$ is
  connected.}
	\end{proof}

\Cref{thm:UpperTh} provides a very simple approach for finding an
$O(\log n)$-approximation of the edge connectivity of a network graph
$G$ in $O(D + \sqrt{n} \log^2 n\log^{*}n)$ rounds, simply by trying exponentially growing sampling probabilities and checking the connectivity.  The proof appears
\fullOnly{in Appendix \ref{app:sampling}}\shortOnly{the full version \cite{Cut-FullVersion}}. We note that a similar basic approach has been used to approximate \emph{the size} of min-cut in the streaming model~\cite{Ahn2012}.

\begin{corollary}\label{crl:approx} There exists a distributed
  algorithm that for any unweighted multi-graph $G=(V,E)$, in $O(D + \sqrt{n} \log^2 n \log^{*}n)$ rounds,
  finds an approximation $\tilde{\C}$ of the edge connectivity such that $\tilde{\C}\in [\C, \C \cdot \Theta(\log n)]$ with high probability.
	\end{corollary}


\section{Min-Cut Approximation by Random Layering}\label{sec:layeringCut}
Now we use random layering to design a min-cut approximation algorithm. We present the outline of the algorithm and its key parts in Subsections \ref{subsec:algoutline} and \ref{subsec:CutTester}. Then, in Subsection \ref{subsec:tie}, we explain how to put these parts together to prove the following theorem:

\begin{theorem} \label{thm:MinCutLayering}
  There is a distributed algorithm that, for any $\epsilon \in (0, 1)$, finds an $O(\epsilon^{-1})$-minimum cut in $O(D)+ O(n^{0.5+\epsilon} \log^3 n \log\log n \log^{*} n)$ rounds, w.h.p.
\end{theorem}

\subsection{Algorithm Outline}\label{subsec:algoutline}
The algorithm is based on closely studying the sampled graph when the edge-sampling probability is
between the two extremes of $\frac{1}{\C}$ and $\frac{\Theta(\log n)}{\C}$. Throughout this process, we identify a set
$\mathcal{F}$ of $O(n \log n)$ cuts such that, with at least a
`reasonably large probability', $\mathcal{F}$ contains at least one
`small' cut.

\smallskip
\begin{mdframed}[hidealllines=true,backgroundcolor=gray!20]
\vspace{-0.2cm}
\paragraph{The Crux of the Algorithm:} Sample edges with probability
$p=\frac{\epsilon \log n}{2\C}$ for a small $\epsilon \in (0, 1)$. Also, assign
each edge to a random layer in $[1, \dots, L]$, where $L = 20\log n$. For each layer $i \in
[1, \dots, L-1]$, let $S_i$ be the set of sampled edges of layer
$i$ and let $S_{i-}=\bigcup_{j=1}^{i} S_j$. For each layer $i \in [1,
\dots, L-1]$, for each component $\mathcal{C}$ of graph $G_i = (V,
S_{i-})$, add the cut $(\mathcal{C}, V\setminus \mathcal{C})$ to the
collection $\mathcal{F}$. \fullOnly{Since in each layer we add at most $n$ new
cuts and there are $L=O(\log n)$ layers, we collect $O(n
\log n)$ cuts in total.} 
\end{mdframed}
\smallskip

We show that with probability at least $n^{-\epsilon}/2$, at least one of the cuts in $\mathcal{F}$
is an $O(\epsilon^{-1})$-minimum cut. Note that thus repeating the
experiment for $\Theta(n^\epsilon \log n)$ times is enough to get that an $O(\epsilon^{-1})$-minimum cut is found w.h.p.

\begin{theorem} \label{thm:main} Consider performing the above
  sampling and layering experiment with edge sampling probability
  $p=\frac{\epsilon\log n}{2\C}$ for $\epsilon \in (0, 1)$ and $L=20\log n$
  layers. Then, $\Pr[\mathcal{F}\textit{ contains an
  }O(\epsilon^{-1})\textit{-minimum cut}]\geq n^{-\epsilon}/2.$
\end{theorem}

\begin{proof}
  Fix an edge sampling probability $p=\frac{\epsilon\log n}{2\C}$ for an
  $\epsilon \in (0, 1)$ and let $\alpha= 40 \epsilon^{-1}$. We say that a
  sampling and layering experiment is \emph{successful} if
  $\mathcal{F}$ contains an $\alpha$-minimum cut or if the sampled
  graph $G_L=(V, S_{L-})$ is connected. We first show that each
  experiment is \emph{successful} with probability at least
  $1-\frac{1}{n}$. The proof of this part is very similar to that of
  \Cref{thm:UpperTh}.

  For an arbitrary layer number $ 1\leq i \leq L-1$, consider graph
  $G_i=(V, S_{i-})$. If $M_i=1$ meaning that $G_i$ is connected, then
  $G_L$ is also connected. Thus, in that case, the experiment is
  successful and we are done. In the more interesting case, suppose
  $M_i \geq 2$. For each component $\mathcal{C}$ of $G_i$, consider
  the cut $(\mathcal{C}, V\setminus \mathcal{C})$. If any of these
  cuts is $\alpha$-minimum, then the experiment is successful as then,
  set $\mathcal{F}$ contains an $\alpha$-minimum cut. On the other
  hand, suppose that for each component $\mathcal{C}$ of $G_i$, we
  have $w(C) \geq \alpha \C$. Then, for each such component
  $\mathcal{C}$, each of the edges of cut $(\mathcal{C}, V\setminus
  \mathcal{C})$ has probability $\frac{p/L}{1-(i\, p)/L} \geq p/L$ to
  be in set $S_{i+1}$ and since $w(\mathcal{C}) \geq \alpha \C$, where
  $\alpha= 20 \epsilon^{-1}$, the probability that none of the edges
  of this cut in set $S_{i+1}$ is at most $(1-p/L)^{\alpha\C} \leq
  e^{\frac{p}{L}\cdot \alpha \C}=e^{-\frac{\epsilon\log n}{2\C}\cdot
    \frac{1}{L}\cdot\frac{40}{\epsilon}\cdot\C} =1/e$. Hence, the probability that component $\mathcal{C}$
  is \emph{bad} as defined in the proof of \Cref{thm:UpperTh} (i.e.,
  in graph $G_{i+1}$, it does not get connected to any other
  component) is at most $1/e$. The rest of the proof can be completed
  exactly as the last paragraph of of the proof of \Cref{thm:UpperTh},
  to show that $$\Pr[\textit{successful experiment}] \geq
  1-{1}/{n}.$$  
	\fullOnly{Thus we have a bound on the probability that
  $\mathcal{F}$ contains an $\alpha$-minimum cut or that the sampled
  graph $G=(V, S_{L-})$ is connected. However, in \Cref{thm:main}, we
  are only interested in the probability of $\mathcal{F}$ containing
  an $\alpha$-minimum cut.}
	Using a union bound, we know that
  \begin{equation}\label{eq:union-of-success}
    \Pr[\textit{successful experiment}] \leq 
    \Pr[\mathcal{F}\textit{ contains an }\alpha\textit{-min cut}] +
    \Pr[G_L\textit{ is connected}]. \nonumber
  \end{equation}
  On the other hand, $$\Pr[G_L\textit{ is connected}] \leq
  1-n^{-\epsilon}.$$ This is because, considering a single mininmum cut of size
  $\C$, the probability that none of the edges of this cut are sampled, in which case the
  sampled subgraph is disconnected, is $(1-\frac{\epsilon \log n}{2\C})^{\C} \geq n^{-\epsilon}$. 
	Hence, we can conclude that $$\quad\,\Pr[\mathcal{F}\textit{ contains an
  }\alpha\textit{-min cut}] \geq (1-{1}/{n}) - (1- n^{-\epsilon}) =
  n^{-\epsilon} - {1}/{n} \geq n^{-\epsilon}/2.$$
\end{proof}

\textbf{Remark:} It was brought to our attention that the approach of \Cref{thm:main} bears some cosmetic resemblance to the technique of Goel, Kapralov and Khanna~\cite{RefinementSampling}. As noted by Kapralov\cite{KapralovPersonalCommunication}, the approaches are fundamentally different; the only similarity is having $O(\log n)$ repetitions of sampling. In \cite{RefinementSampling}, the objective is to estimate the \emph{strong-connectivity} of edges via a streaming algorithm. See \cite{RefinementSampling} for related definitions and note also that strong-connectivity is (significantly) different from (standard) connectivity. In a nutshell, \cite{RefinementSampling} uses $O(\log n)$ iterations of sub-sampling, each time further sparsifying the graph until at the end, all edges with strong-connectivity less than a threshold are removed (and identified) while edges with strong connectivity that is a $\Theta(\log n)$ factor larger than the threshold are preserved (proven via Benczur-Karger's sparsification).  

\subsection{Testing Cuts}\label{subsec:CutTester}
So far we know that $\mathcal{F}$ contains an $\alpha$-minimum cut
with a reasonable probability. We now need to devise a distributed
algorithm to read or test the sizes of the cuts in $\mathcal{F}$ and
find that $\alpha$-minimum cut, in $O(D)+\tilde{O}(\sqrt{n})$
rounds. \fullOnly{In the remainder of this section, we explain our approach to
this part.}

Consider a layer $i$ and the graph $G_i = (V,
S_{i-})$. \fullOnly{For each
component $\mathcal{C}$ of $G_i$, $diam(\mathcal{C})$ rounds is enough to
read the size of the cut $(\mathcal{C}, V\setminus \mathcal{C})$ such
that all the nodes in component $\mathcal{C}$ know this size. However,
$diam(\mathcal{C})$ can be considerably larger than $D = diam(G)$ and
thus, this method would not lead to a round complexity of
$\tilde{O}(D+\sqrt{n})$. To overcome this problem, n}\shortOnly{N}otice that we do
not need to read the exact size of the cut $(\mathcal{C}, V\setminus \mathcal{C})$. Instead, it is enough to
devise a \emph{test} that \emph{passes} w.h.p.\ if $w(C) \leq
\alpha\C$, and \emph{does not pass} w.h.p. if $w(C) \geq
(1+\delta)\alpha \C$, for a small constant $\delta \in (0, 1/4)$. In
the distributed realization of such a test, it would be enough if all
the nodes in $\mathcal{C}$ consistently know whether the test passed
or not. Next, we explain a simple algorithm for such a test. This test
itself uses random edge sampling. Given such a test, in each layer $i
\in [1, \dots, L-1]$, we can test all the cuts and if any cut passes
the test, meaning that, w.h.p., it is a $((1+\delta)\alpha)$-minimum
cut, then we can pick such a cut.\footnote{This can be done for
  example by picking the cut which passed the test and for which the
  related component has the smallest id among all the cuts that passed
  the test.}

\begin{lemma}\label{lem:CutTester} Given a subgraph $G'=(V, E')$ of
  the network graph $G=(V,E)$, a threshold $\kappa$ and $\delta \in (0,1/4)$, there exists a randomized distributed
  cut-tester algorithm with round complexity
  $\Theta\big(D+\frac{1}{\delta^2}\sqrt{n} \log n \log^* n\big)$ such
  that, w.h.p., for each node $v\in V$, we have: Let
  $\mathcal{C}$ be the connected component of $G'$ that contains
  $v$. If $w(\mathcal{C}) \leq \kappa /(1+\delta)$, the test passes at
  $v$, whereas if $w(\mathcal{C}) \geq \kappa (1+\delta)$, the test
  does not pass at $v$.
\end{lemma}

For pseudo-code, we refer to \fullOnly{\Cref{App:LayeringCut}}\shortOnly{the full
version \cite{Cut-FullVersion}}. We first run Thurimella's connected
component identification algorithm (refer to \Cref{subsec:Thurimella})
on graph $G$ for subgraph $G'$, so that each node $v\in V$ knows the
smallest id in its connected component of graph $G'$. Then, each node
$v$ adopts this label $componentID$ as its own id (temporarily). Thus,
nodes of each connected component of $G'$ will have the same id. Now,
the test runs in $\Theta(\log^2 n/\delta^2)$ \emph{experiments}, each
as follows: in the $j^{th}$ experiment, for each edge $e \in E
\setminus E'$, put edge $e$ in set $E_j$ with probability
$p'=1-2^{-\,\frac{1}{\kappa}}$. Then, run Thurimella's algorithm on
graph $G$ with subgraph $H_j=(V, E' \cup E_j)$ and with the new ids
twice, such that at the end, each node $v$ knows the smallest and the
largest id in its connected component of $H_j$. Call these new labels
$\ell^{min}_j(v)$ and $\ell^{max}_j(v)$, respectively. For a node $v$
of a component $\mathcal{C}$ of $G_i$, we have that $\ell^{min}_j(v)
\neq v.id$ or $\ell^{max}_j(v) \neq v.id$ iff at least one of the
edges of cut $(\mathcal{C}, V\setminus \mathcal{C})$ is sampled in
$E_j$, i.e., $(\mathcal{C}, V\setminus \mathcal{C}) \cap E_j \neq
\emptyset$. Thus, each node $v$ of each component $\mathcal{C}$ knows
whether $(\mathcal{C}, V\setminus \mathcal{C}) \cap E_j \neq
\emptyset$ or not. Moreover, this knowledge is consistent between all
the nodes of component $\mathcal{C}$. After $\Theta(\log n/\delta^2)$
experiments, each node $v$ of component $\mathcal{C}$ considers the
test \emph{passed} iff $v$ noticed $(\mathcal{C}, V\setminus
\mathcal{C}) \cap E_j \neq \emptyset$ in at most half of the
experiments. We defer the calculations of the proof of
\Cref{lem:CutTester} to \fullOnly{\Cref{App:LayeringCut}}\shortOnly{of
  the full version
\cite{Cut-FullVersion}}.

\subsection{Putting the Pieces Together}\label{subsec:tie}
We now explain how to put together the pieces presented in the previous subsections to get the claim of \Cref{thm:MinCutLayering}.
\begin{proof}[Proof of \Cref{thm:MinCutLayering}]
For simplicity, we first explain an $O(\eps^{-1})$ minimum-cut approximation algorithm with time complexity $O((D+\sqrt{n}\log^{*}n \log n) n^{\epsilon} \log^2n \log\log n)$. Then, explain how to reduce it to the claimed bound of $O(D) + O(n^{0.5+\epsilon} \log^{3}n \log\log n \log^* n)$ rounds.

We first find an $O(\log n)$ approximation $\tilde{\C}$ of $\C$, using \Cref{crl:approx}, in time $O(D) + O(\sqrt{n}\log^{*})\log^2 n)$. This complexity is subsumed by the complexity of the later parts. After this, we use $\Theta(\log \log n)$ guesses for a $2$-approximation of $\C$ in the form $\C'_i=\tilde{C} 2^{i}$ where $i \in [-\Theta(\log \log n), \Theta(\log\log n)]$. For each such guess $\C'_i$, we have $n^{\epsilon} \log n$ \emph{epochs} as follows:

In each \emph{epoch}, we sample edges with probability $p=\frac{\epsilon \log n}{2\C'}$ and assign each edge to a random layer in $[1, \dots, L]$, where $L=20\log n$. For each layer $i \in [1, \dots, L-1]$, we let $S_i$ be the set of sampled edges of layer $i$ and let $S_{i-}=\cup_{j=1}^{i} S_j$. Then, for each $i \in [1, \dots, L]$, we use the Cut-Tester Algorithm (see \Cref{subsec:CutTester}) on graph $G$ with subgraph $G_{i}= (V, S_{i-})$, threshold $\kappa = 50\C'/\epsilon$, and with parameter $\delta=1/8$. This takes $O((D+\sqrt{n}\log n\log^{*}n) \log n)$ rounds (for each layer). If in a layer, a component passes the test, it means its cut has size at most $O(\C'/\epsilon)$, with high probability. To report the results of the test, we construct a BFS tree rooted in a leader in $O(D)$ rounds and we convergecast the minimum $componentID$ that passed the test, in time $O(D)$. We then broadcast this $componentID$ to all nodes and all nodes that have this $componentID$ define the cut that is $O(\C'/\epsilon)$-minimum, with high probability. 

Over all the guesses, we know that there is a guess $\C'_j$ that is a $2$-approximation of $\C$. In that guess, from \Cref{thm:main} and a Chernoff bound, we know that at least one cut that is an $O(\epsilon^{-1})$-minimum cut will pass the test. We stop the process in the smallest guess for which a cut passes the test. 

Finally, to reduce the time complexity to $O(D) + O(n^{0.5+\epsilon} \log^{3}n \log\log n \log^* n)$ rounds, note that we can parallelize (i.e., pipeline) the $\Theta(n^{\epsilon} \log^2 n \log \log n)$ runs of Cut-Testing algorithm, which come from $\Theta(\log\log n)$ guesses $\C'_i$, $n^{\epsilon} \log n$ epochs for each guess, and $\Theta(\log n)$ layers in each epoch. We can do this pipe-lining simply because these instances of Cut-Testing do not depend on the outcomes of each other and $k$ instances of Thurimella's algorithms can be run together in time $O(D+k\sqrt{n}\log^* n)$ rounds (refer to \Cref{subsec:Thurimella}). To output the final cut, when doing the convergecast of the Cut-Testing results on the BFS, we append the edge-connectivity guess $\C'_j$, epoch number, and layer number to the $componentID$. Then, instead of taking minimum on just $componentID$, we choose the $componentID$ that has the smallest tuple  (guess $\C'_j$, epoch number, layer number, $componentID$). Note that the smallest guess $\C'_j$ translates to the smallest cut size, and the other parts are simply for tie-breaking.
\end{proof}


\section{Min-Cut Approximation via Matula's Approach}\label{sec:Matula}
In \cite{Matula93}, Matula presents an elegant centralized algorithm
that for any constant $\eps>0$, finds a $(2+\eps)$-min-cut in
$O(|V|+|E|)$ steps. Here, we explain how with
the help of a few additional elements, this general approach can be
used in the distributed setting, to find a $(2+\eps)$-minimum cut in
$O\big((D+ \sqrt{n} \log^{*} n) \log^2 n \log \log n \cdot
\frac{1}{\eps^5}\big)$ rounds. We first recap the concept of
\emph{sparse certificates for edge connectivity}.

\begin{definition} For a given unweighted multi-graph $H=(V_H, E_H)$ and a
  value $k>0$, a set $E^* \subseteq E_H$ of edges is a \emph{sparse
    certificate for $k$-edge-connectivity} of $H$ if (1) $|E^*| \leq k |V_{H}|$,
  and (2) for each edge $e\in E_H$, if there exists a cut $(\mathcal{C},
  V\setminus \mathcal{C})$ of $H$ such that $|(\mathcal{C})|\leq k$ and $e
  \in(\mathcal{C}, V\setminus \mathcal{C})$, then we have $e\in E^*$.
\end{definition}

Thurimella~\cite{Thurimella97} presents a simple distributed algorithm
that finds a sparse certificate for $k$-edge-connectivity of a network
graph $G$ in $O(k(D+ \sqrt{n} \log^{*} n))$ rounds. With simple
modifications, we get a generalized version, presented in
\Cref{lem:ThurimellaCertificate}. Details of these modification appear
in \fullOnly{\Cref{App:Matula}}\shortOnly{the full version of this paper
  \cite{Cut-FullVersion}}.

\begin{lemma}\label{lem:ThurimellaCertificate} Let $E_c$ be a subset of the edges of the network graph $G$ and define the virtual graph $G'=(V', E')$ as the multi-graph that is obtained by contracting all the edges of $G$ that are in $E_c$. Using the modified version of Thurimella's certificate algorithm, we can find a set $E^* \subseteq E \setminus E_c$ that is a sparse certificate for $k$-edge-connectivity of $G'$, in $O(k(D+ \sqrt{n} \log^{*} n))$ rounds.
\end{lemma}
  Following the approach of Matula's centralized algorithm\footnote{We remark that
    Matula~\cite{Matula93} never uses the name \emph{sparse
      certificate} but he performs \emph{maximum adjacency search}
    which indeed generates a sparse certificate.}~\cite{Matula93}, and
  with the help of the sparse certificate algorithm of
  \Cref{lem:ThurimellaCertificate} and the random sparsification technique of Karger~\cite{KargerSTOC94}, we get the following result.
		
\begin{theorem}\label{thm:distributedmatula}
  There is a distributed algorithm that, for any constant $\eps>0$,
  finds a $(2+\eps)$-minimum cut in $O((D+ \sqrt{n} \log^{*} n) \log^2
  n \log \log n \cdot \frac{1}{\eps^5})$ rounds.
\end{theorem}

\begin{algorithm}[H]
\caption{$(2+\eps)$-minimum cut approximation: Matula's Approach}
\begin{algorithmic}[1]
\algorithmsize
	\Statex Given a $(1+\eps/10)$-factor approximation $\tilde{\C}$ of $\C$
	\Statex
	\State $E_c \gets \emptyset$, $E^* \gets E$, $\eta_{old} \gets n$, $\eta_{new} \gets 1$
	\While{($\eta\geq 2$) \& ($\eta_{new} \leq \eta_{old} (1-\eps/10)$)} 
			\State $E_c \gets E \setminus E^*$
			\State $E^* \gets$ a sparse certificate for $\tilde{\C}(1+\eps/5)$-edge-connectivity of graph $G'=(V', E')$ obtained by contracting edges of $E_c$
			\State $\eta_{new} \gets$ number of connected components of subgraph $H=(V, E\setminus E^*)$
	\EndWhile
	\State\textbf{endwhile}
	\State Test cuts defined by connected components of graph $H=(V, E\setminus E^*)$ versus threshold $\kappa=\tilde{\C}(2+\eps/3)$
	\State Output the component that passes the test and contains the smallest id between such components
\label{alg:Matula}
\end{algorithmic}
\end{algorithm}

\begin{proof}[Proof of \Cref{thm:distributedmatula}]
  We assume that nodes know a $(1+\eps/10)$-factor approximation
  $\tilde{\C}$ of the edge connectivity $\C$, and explain a
  distributed algorithm with round complexity $O((D+ \sqrt{n} \log^{*}
  n) \log^2 n \cdot \frac{1}{\eps^4})$. Note that this assumption can
  be removed at the cost of a $\Theta(\frac{\log \log n}{\log
    {(1+\eps/10)}}) = \Theta(\log \log n \cdot \frac{1}{\eps})$ factor
  increase in round complexity by trying $\Theta(\frac{\log\log
    n}{\eps})$ exponential guesses $\tilde{\C}(1+\eps/10)^i$ for $i\in
  [0, \Theta(\frac{\log\log n}{\eps})]$
  where $\tilde{\C}$ is an $O(\log n)$-approximation of the
  edge-connectivity, which can be found by \Cref{crl:approx}.
 
  For simplicity, we first explain an algorithm that finds a
  $(2+\eps)$-minimum cut in $O(\lambda (D+ \sqrt{n} \log^{*} n) \log n
  \cdot \frac{1}{\eps^2})$ rounds. Then, we explain how to reduce the
  round complexity to $O((D+ \sqrt{n} \log^{*} n) \log^2 n \cdot
  \frac{1}{\eps^4})$. A pseudo-code is presented in Algorithm \ref{alg:Matula}.

  First, we compute a sparse certificate $E^*$ for
  $\tilde{\C}(1+\eps/5)$-edge-connectivity for $G$, using Thurimella's
  algorithm. Now consider the graph $H=(V, E\setminus E^*)$. We have
  two cases: either (a) $H$ has at most $|V|(1-\eps/10)$ connected
  components, or (b) there is a connected component $\mathcal{C}$ of
  $H$ such that $w(\mathcal{C}) \leq
  \frac{2\C(1+\eps/10)(1+\eps/5)}{1-\eps/10} \leq (2+\eps) \C$. Note
  that if (a) does not hold, case (b) follows because $H$ has at most
  $(1+\eps/5)\tilde{\C}|V|$ edges.
	
  In Case (b), we can find a $(2+\eps)$-minimum cut by testing the
  connected components of $H$ versus threshold
  $\kappa=\tilde{\C}(2+\eps/3)$, using the Cut-Tester algorithm
  presented in \Cref{lem:CutTester}.
  In Case (a), we can solve the problem recursively on the virtual
  graph $G'=(V', E')$ that is obtained by contracting all the edges of
  $G$ that are in $E_c= E \setminus E^*$. Note that this contraction
  process preserves all the cuts of size at most $\tilde{\C}(1+\eps/5)
  \geq \C$ but reduces the number of nodes (in the virtual graph) at
  least by a $(1-\eps/10)$-factor. Consequently, $O(\log(n)/\eps)$
  recursions reduce the number of components to at most $2$ while
  preserving the minimum cut.

  The dependence on $\C$ can be removed by considering the graph
  $G_S=(V,E_S)$, where $E_S$ independently contains every edge of $G$
  with probability $\Theta\big(\frac{\log n}{\eps^2\C}\big)$. It can
  be shown that the edge connectivity of $G_S$ is
  $\Theta(\log(n)/\eps^2)$ and a minimum edge cut of $G_S$ gives a
  $(1+O(\eps))$-minimum edge cut of $G$. 
  
  We now explain how to remove the dependence on $\C$ from the time
  complexity. Let $E_S$ be a subset of the edges of $G=(V,E)$ where
  each $e\in E$ is independently included in $E_S$ with probability
  $p=\frac{100\log n}{\eps^2} \cdot \frac{1}{\C}$. Then, using the
  edge-sampling result of Karger~\cite[Theorem
  2.1]{KargerSTOC94}\footnote{We emphasize that this result is
    non-trivial. The proof follows from the powerful bound of
    $O(n^{2\alpha})$ on the number of $\alpha$-minimum
    cuts~\cite{Karger93-Contraction} and basic concentration arguments
    (Chernoff and union bounds).}, we know that with high probability,
  for each $\mathcal{C} \subseteq V$, we have
  $$(1-\eps/3)\cdot|(\mathcal{C}, V\setminus \mathcal{C})|\cdot p \leq
  |(\mathcal{C}, V\setminus \mathcal{C}) \cap E_S| \leq
  (1+\eps/3)\cdot|(\mathcal{C}, V\setminus \mathcal{C})|\cdot p.$$
  Hence, in particular, we know that graph $G_{new}=(V, E_S)$ has edge
  connectivity at least $\C p (1-\eps/3)$ and at most $\C p
  (1+\eps/3)$, i.e., $\C_{new}=\Theta(\log n \cdot
  \frac{1}{\eps^2})$. Moreover, for every cut $(\mathcal{C},
  V\setminus \mathcal{C})$ that is a $(1+\eps/3)$-minimum cut in graph
  $G_{new}$, we have that $(\mathcal{C}, V\setminus \mathcal{C})$ is a
  $(1+\eps)$-minimum cut in graph $G$. We can therefore solve the
  cut-approximation problem in graph $G_{new}$, where we only need to
  use sparse certificates for $\Theta(\log n \cdot \frac{1}{\eps^2})$
  edge-connectivity\footnote{Note that, solving the cut approximation
    on the virtual graph $G_{new}$ formally means that we set the
    weight of edges outside $E\setminus E_0$ equal to zero. However,
    we still use graph $G$ to run the distributed algorithm and thus,
    the round complexity depends on $diam(G)=D$ and not on the possibly
    larger $diam(G_{new})$.}. The new round complexity becomes $O\big((D+
  \sqrt{n} \log^{*} n) \log^2 n \cdot \frac{1}{\eps^4}\big)$ rounds. 
	
	The above round complexity is assuming a $(1+\eps/10)$-approximation of edge-connectivity is known. Substituting this assumption with trying $\Theta(\log \log n /\eps)$ guesses around the $O(\log n)$ approximation obtained by \Cref{crl:approx} (and outputting the smallest found cut) brings the round complexity to the claimed bound of $$O((D+ \sqrt{n} \log^{*} n) \log^2
  n \log \log n \cdot \frac{1}{\eps^5}).$$ 
\end{proof}

\section{Lower Bounds}
\label{sec:lowerbounds}


In this section, we present our lower bounds for minimum cut approximation, which can be viewed as strengthening and
generalize some of the lower bounds of Das Sarma et al.\ from
\cite{dassarma12}.

Our lower bound uses the same general approach as the lower bounds in
\cite{dassarma12}. The lower bounds of \cite{dassarma12} are based on
an $n$-node graph $G$ with diameter $O(\log n)$ and two distinct nodes
$s$ and $r$. The proof deals with distributed protocols where node $s$
gets a $b$-bit input $x$, node $r$ gets a $b$-bit input $y$, and apart
from $x$ and $y$, the initial states of all nodes are globally
known. Slightly simplified, the main technical result of
\cite{dassarma12} (Simulation Theorem 3.1) states that if there is a
randomized distributed protocol that correctly computes the value
$f(x,y)$ of a binary function $f: \set{0,1}^b\times\set{0,1}^b\to
\set{0,1}$ with probability at least $1-\eps$ in time
$T$ (for sufficiently small $T$), then there is also a randomized
$\eps$-error two-party protocol for computing $f(x,y)$ with
communication complexity $O(TB\log n)$. For our lower bounds, we need
to extend the simulation theorem of \cite{dassarma12} to a larger
family of networks and to a slightly larger class of problems. 

\subsection{Generalized Simulation Theorem}
\label{sec:simulation}

\paragraph{Distributed Protocols:} Given a weighted network graph
$G=(V,E,w)$ ($\forall e\in E: w(e)\geq 1$), we consider distributed
tasks for which each node $v\in V$ gets some private input $x(v)$ and
every node $v\in V$ has to compute an output $y(v)$ such that the
collection of inputs and outputs satisfies some given
specification. To solve a given distributed task, the nodes of $G$
apply a distributed protocol. We assume that initially, each node
$v\in V$ knows its private input $x(v)$, as well as the set of
neighbors in $G$. Time is divided into synchronous rounds and in each
round, every node can send at most $B\cdot w(e)$ bits over each of its
incident edges $e$. We say that a given (randomized) distributed
protocol solves a given distributed task with error probability $\eps$
if the computed outputs satisfy the specification of the task with
probability at least $1-\eps$.

\paragraph{Graph Family {\boldmath$\mathcal{G}(n,k,c)$}:}
For parameters $n$, $k$, and $c$, we define the \emph{family of
  graphs} $\mathcal{G}(n,k,c)$ as follows. A weighted graph
$G=(V,E,w)$ is in the family $\mathcal{G}(n,k,c)$ iff
$V=\set{0,\dots,n-1}$ and for all $h\in \set{0,\dots,n-1}$, the total
weight of edges between nodes in $\set{0,\dots,h}$ and nodes in
$\set{h+k+1,\dots,n-1}$ is at most $c$. We consider distributed
protocols on graphs $G\in \mathcal{G}(n,k,c)$ for given $n$, $k$, and
$c$. For an integer $\eta\geq 1$, we define
$L_\eta:=\set{0,\dots,\eta-1}$ and $R_\eta:=\set{n-\eta,\dots,n-1}$.

\smallskip
Given a parameter $\eta\geq 1$ and a network $G\in
\mathcal{G}(n,k,c)$, we say that a two-party protocol between Alice
and Bob $\eta$-solves a given distributed task for $G$ with error
probability $\eps$ if a) initially Alice knows all inputs and all
initial states of nodes in $V\setminus R_\eta$ and Bob knows all
inputs and all initial states of nodes in $V\setminus L_\eta$, and b)
in the end, Alice outputs $y(v)$ for all $v\in L_{n/2}$ and Bob
outputs $y(v)$ for all $v\in R_{n/2}$ such that with probability at
least $1-\eps$, all these $y(v)$ are consistent with the specification
of the given distributed task. A two-party protocol is said to be
\emph{public coin} if Alice and Bob have access to a common random
string. 

\begin{theorem}[Generalized Simulation Theorem]\label{thm:simulation}
  Assume we are given positive integers $n$, $k$, and $\eta$, a
  parameter $c\geq 1$, as well as a subfamily
  $\tilde{\mathcal{G}}\subseteq \mathcal{G}(n,k,c)$. Further assume
  that for a given distributed task and graphs $G\in
  \tilde{\mathcal{G}}$, there is a randomized protocol with error
  probability $\eps$ that runs in $T \leq (n-2\eta)/(2k)$
  rounds. Then, there exists a public-coin two-party protocol that
  $\eta$-solves the given distributed task on graphs $G\in
  \tilde{\mathcal{G}}$ with error probability $\eps$ and communication
  complexity at most $2Bc T$.
\end{theorem}
\begin{proof}
  We show that Alice and Bob can simulate an execution of the given
  distributed protocol to obtain outputs that are consistent with the
  specification of the given distributed task. First note that a
  randomized distributed algorithm can be modeled as a deterministic
  algorithm where at the beginning, each node $v$ receives s
  sufficiently large random string $r(v)$ as additional input. Assume
  that $R$ is the concatenation of all the random strings
  $r(v)$. Then, a randomized distributed protocol with error
  probability $\eps$ can be seen as a deterministic protocol that
  computes outputs that satisfy the specification of the given task
  with probability at least $1-\eps$ over all possible choices of
  $R$. (A similar argument has also been used, e.g., in
  \cite{dassarma12}).

  Alice and Bob have access to a public coin giving them a common
  random string of arbitrary length. As also the set of nodes
  $V=\set{0,\dots,n-1}$ of $G$ is known, Alice and Bob can use the
  common random string to model $R$ and to therefore consistently
  simulate all the randomness used by all $n$ nodes in the distributed
  protocol. Given $R$, it remains for Alice and Bob to simulate a
  deterministic protocol. If they can (deterministically) compute the
  outputs of some nodes of a given deterministic protocol, they can
  also compute outputs for a randomized protocol with error
  probability $\eps$ such that the outputs are consistent with the
  specification of the distributed task with probability at least
  $1-\eps$.

  Given a deterministic distributed protocol on a graph $G\in
  \tilde{\mathcal{G}}$ with time complexity $T\leq (n-2\eta)/(2k)$, we
  now describe a two-party protocol with communication complexity at
  most $2BcT$ in which for each round $r\in \set{0,\dots,T}$,
  \begin{itemize}
  \item[(I)] Alice computes the states of all nodes $i< n-\eta-r\cdot k$
    at the end of round $r$, and
  \item[(II)] Bob computes the states of all nodes $i\geq\eta+r\cdot
    k$ at the end of round $r$.
  \end{itemize}
  Because the output $y(u)$ of every node $u$ is determined by $u$'s
  state after $T$ rounds, together with the upper bound on $T$, (I)
  implies that Alice can compute the outputs of all nodes $i< n/2$
  and Bob can compute the outputs of all nodes $i\geq n/2$. Therefore,
  assuming that initially, Alice knows the states of node $i<
  n-\eta$ and Bob knows the states of nodes $i\geq\eta$, a two-party
  protocol satisfying (I) and (II) $\eta$-solves the distributed task
  solved by the given distributed protocol. In order to prove the
  claim of the theorem, it is thus sufficient to show that there
  exists a deterministic two-party protocol with communication
  complexity at most $2Bc T$ satisfying (I) and (II).

  In a deterministic algorithm, the state of a node $u$ at the end of
  a round $r$ (and thus at the beginning of round $r+1$) is completely
  determined by the state of $u$ at the beginning of round $r$ and by
  the messages node $u$ receives in round $r$ from its neighbors. We
  prove I) and II) by induction on $r$. First note that (interpreting
  the initial state as the state after round $0$), (I) and (II) are
  satisfied by the assumption that initially, Alice knows the initial
  states of all nodes $0,\dots,n-1-\eta$ and Bob knows the initial
  states of all nodes $\eta,\dots,n-1$. Next, assume that (I) and (II)
  hold for some $r=r'\in\set{0,\dots,T-1}$. Based on this, we show how
  to construct a protocol with communication complexity at most $2B c$
  such that (I) and (II) hold for $r=r'+1$. We formally show how,
  based on assuming I) and (II) for $r=r'$, Alice can compute the
  states of nodes $i< n-\eta-(r'+1)k$ using only $Bc$ bits of
  communication. The argument for Bob can be done in a completely
  symmetric way so that we get a total communication complexity of
  $2Bc$. In order to compute the state of a node $i< n-\eta-(r'+1)k$
  at the end of round $r'+1$, Alice needs to know the state of node
  $i$ at the beginning of round $r'+1$ (i.e., at the end of round
  $r'$) and the message sent by each neighbor $j$ in round
  $r'+1$. Alice knows the state of $i$ at the beginning of round $r'$
  and the messages of neighbors $j< n-\eta-r'k$ by the assumption
  (induction hypothesis) that Alice already knows the states of all
  nodes $i< n-\eta-r'k$ at the end of round $r'$. By the definition of
  the graph family $\mathcal{G}(n,k,c)$, the total weight of edges
  between nodes $i< n-\eta-(r'+1)k$ and nodes $j\geq n-\eta-r'k$ is at
  most $c$. The number of bits sent over these edges in round $r'+1$
  is therefore at most $cB$. If at the beginning of round $r'$, Bob
  knows the states of all nodes $j\geq n-\eta-r'k$, Bob can send these
  $cB$ bits to Alice. By assuming that also (II) holds for $r=r'$, Bob
  knows the states of all nodes $j'\geq\eta+r'k$. We therefore need to
  show that $\eta+r'k \leq n-\eta-r'k$ and thus $r'\leq
  (n-2\eta)/(2k)$. Because $r'\leq T-1$, this directly follows from
  the upper bound on $T$ given in the theorem statement and thus, (I)
  and (II) hold for all $r\in \set{0,\dots,T}$.
\end{proof}

\subsection{Lower Bound for Approximating Minimum Cut: Weighted Graphs}
\label{sec:mincutlower}

In this subsection, we prove a lower bound on approximating the minimum cut
in weighted graphs (or equivalently in unweighted
multigraphs). The case of simple unweighted graphs is addressed in the next subsection. 

Let $k\geq 1$ be an integer
parameter. We first define a fixed $n$-node graph $H=(V,E_H)$ that we
will use as the basis for our lower bound. The node set $V$ of $H$ is
$V=\set{0,\dots,n-1}$. For simplicity, we assume that $n$ is an
integer multiple of $k$ and that $\ell:=n/k$. The edge $E_H$ consists
of three parts $E_{H,1}$, $E_{H,2}$, and $E_{H,3}$ such that
$E_H=E_{H,1}\cup E_{H,2}\cup E_{H,3}$. The three sets are defined as
follows.
\begin{eqnarray*}
  E_{H,1} & := & \set{\set{i,j}:i,j\in\{0,\dots,n-1\}\text{ and } j=i+k},\\
  E_{H,2} & := & \set{\set{i,j}:i,j\in\{0,\dots,n-1\}\text{ and }
    \exists s\in\mathbb{N}\text{ s.t.}\ 
    i\equiv 0\!\!\!\pmod{k2^s} \text{ and } j=i+k2^s},\\
  E_{H,3} & := & \set{\set{i,j}:i,j\in\{0,\dots,n-1\},\
    i\equiv 0\!\!\!\pmod{k},\text{ and } 0<j-i\leq k-1}.
\end{eqnarray*}
The edges $E_{H,1}$ connect the nodes $V$ of $H$ to $k$ disjoint paths
of length $\ell$, where for each integer $x\in\set{0,\dots,k-1}$, the
nodes $i\equiv x\!\pmod{k}$ form on of the paths. Using the edges of
$E_{H,2}$, the nodes of the first of these paths are connected to a
graph of small diameter. Finally, using the edges $E_{H,3}$ the paths
are connected to each other in the following way. We can think of the
$n$ nodes as consisting of groups of size $k$, where corresponding
nodes of each of the $k$ paths form a group (for each integer $h\geq
0$, nodes $hk,\dots,(h+1)k-1$ form a group). Using the edges of
$E_{H,3}$ each such group is connected to a star, where the node of
the first path is the center of the star.

Based on graph $H$, we define a family $\mathcal{H}(n,k)$ of weighted
graphs as follows. The family $\mathcal{H}(n,k)$ contains all weighted
versions of graph $H$, where the weights of all edges of $E_{H,2}$ are
$1$ and weights of all remaining edges are at least $1$, but otherwise
arbitrary. The following lemma shows that $\mathcal{H}(n,k)$ is a
subfamily of $\mathcal{G}(n,k,c)$ for appropriate $c$ and that graphs
in $\mathcal{H}(n,k)$ have small diameter.

\begin{lemma}\label{lemma:lowerboundgraph}
  We have $\mathcal{H}(n,k)\subset \mathcal{G}(n,k,c)$ for $c=
  \log_2(n/k)$. Further, each graph in $\mathcal{H}(n,k)$ has
  diameter at most $O(\log(n/k))$.
\end{lemma}
\begin{proof}
  To show that $\mathcal{H}(n,k)\subset\mathcal{G}(n,k,c)$, we need to
  show that for each $h\in\set{0,\dots,n-1}$, the total weight of
  edges between nodes in $\set{0,\dots,h}$ and nodes in
  $\set{h+k+1,\dots,n-1}$ is at most $c$. All edges in $E_{H,1}$ and
  $E_{H,3}$ are between nodes $i$ and $j$ for which $|j-i|\leq k$, the
  only contribution to the weight of edges between nodes in
  $\set{0,\dots,h}$ and nodes in $\set{h+k+1,\dots,n-1}$ thus comes
  from edges $E_{H,2}$. For each $h\in\set{0,\dots,n-1}$ and for each
  $s\in \mathbb{N}$, there is at most one pair $(i,j)$ such that
  $i\equiv j\equiv \!\pmod{k2^s}$ and such that $i\leq h$ and
  $j>h+k$. The number of edges between nodes in $\set{0,\dots,h}$ and nodes in
  $\set{h+k+1,\dots,n-1}$ therefor is at most $\log_2(n/k)$ and the
  first claim of the lemma therefore follows because edges in
  $E_{H,2}$ are required to have weight $1$. The bound on the diameter
  follows directly from the construction: With edges $E_{H,3}$, each
  node is directly connected to a node of the first path and with
  edges $E_{H,2}$, the nodes of the first path are connected to a
  graph of diameter $O(\log(n/k))$.
\end{proof}

Based on the graph family $\mathcal{H}(n,k)$ as defined above, we can
now use the basic approach of \cite{dassarma12} to prove a lower bound
for the distributed minimum cut problem.

\begin{theorem}\label{thm:cutlowerweighted}
  In weighted graphs (and unweighted multi-graphs), for any
  $\alpha\geq1$, computing an $\alpha$-approximation of the edge
  connectivity $\lambda$ or computing an $\alpha$-approximate minimum
  cut (even if $\lambda$ is known) requires at least $\Omega\big(D +
  \sqrt{n/(B\log n)}\big)$ rounds, even in graphs of diameter
  $D=O(\log n)$.
\end{theorem}
\begin{proof}
  We prove the theorem by reducing from the two-party set disjointness
  problem \cite{chattapodhyay11,kalyanasundaram92,razborov92}. Assume
  that as input, Alice gets a set $X$ and Bob get a set $Y$ such that
  the elements of $X$ and $Y$ are from a universe of size $O(p)$. It
  is known that in general, Alice and Bob need to exchange at least
  $\Omega(p)$ bits in order to determine whether $X$ and $Y$ are
  disjoint \cite{kalyanasundaram92,razborov92}. This lower bound holds
  even for public coin randomized protocols with constant error
  probability and it also holds if Alice and Bob are given the promise
  that if $X$ and $Y$ intersect, they intersect in exactly one element
  \cite{razborov92}. As a consequence, if Alice and Bob receive sets
  $X$ and $Y$ as inputs with the promise that $|X\cap Y|=1$, finding
  the element in $X\cap Y$ also requires them to exchange $\Omega(p)$
  bits.

  Assume that there is a protocol to find an $\alpha$-minimum cut or
  to $\alpha$-approximate the size of a minimum cut in time $T$ with a
  constant error probability $\eps$. In both cases, if $T$ is sufficiently small, we show that Alice
  and Bob can use this protocol to efficiently solve set disjointness
  by simulating the distributed protocol on a special network from the
  family $\mathcal{H}(n,k)$.

  We now describe the construction of this network
  $G\in\mathcal{H}(n,k)$. We assume that the set disjointness inputs
  $X$ and $Y$ of Alice and Bob are both of size $\Theta(k)$ and from a
  universe of size $k-1$. The structure of $G$ is already given, the
  edge weights of edges in $E_{H,1}$ and $E_{H,3}$ are given as
  follows. First, all edges $E_{H,1}$ (the edges of the paths) have
  weight $\alpha\ell+1$ (recall that $\ell=n/k$ is the length of the
  paths). We number the paths from $0$ to $k-1$ as follows. Path
  $p\in\set{0,\dots,k-1}$ consists of all nodes $i$ for which $i\equiv
  p\!\pmod{k}$. Note that the first node of path $p$ is node $p$ and
  the last node of path $p$ is $n-k+p$. We encode the set disjointness
  inputs $X$ and $Y$ in the edge weights of the edges of $E_{H,3}$ as
  follows. For each $x\in\set{0,\dots,k-1}\setminus X$, the edge
  between node $0$ and node $x$ has weight $\alpha\ell+1$. Further,
  for each $y\in\set{0,\dots,k-1}\setminus Y$, the edge between $n-k$
  and $n-k+y$ has weight $\alpha\ell+1$. All other edges of $E_{H,3}$
  have weight $1$.

  Hence, the graph induced by the edges with large weight $\alpha\ell
  + 1$ (in the following called heavy edges) looks as follows. It
  consists of the $k$ paths of length $\ell$. In addition for each
  $x\not\in X$, path $x$ is connected to node $0$ and for each
  $y\not\in Y$, path $y$ is connected to node $n-k$. Assume that there
  is exactly one element $z\in X\cap Y$. Path $z$ is not connected to
  path $0$ through a heavy edge, all other paths are connected to
  each other by heavy edges. The minimum cut $(S,V\setminus S)$ is
  defined by the nodes $S=\set{i\in\set{0,\dots,n-1}:i\equiv
    z\!\pmod{k}}$. As each node on path $z$ is connected by a single
  weight $1$ edge to a node on path $0$, the size of the cut
  $(S,V\setminus S)$ is $\ell$. There is at least one heavy edge
  crossing every other cut and thus, every other cut has size at least
  $\alpha\ell+1$. In order to find an $\alpha$-approximate minimum
  cut, a distributed algorithm therefore has to find path $z$ and thus
  the element $z\in X\cap Y$.

  Assume now that there is a distributed protocol that computes an
  $\alpha$-approximate minimum cut in $T$ rounds, by using messages of
  at most $B$ bits. The described graph $G$ is in $\mathcal{H}(n,k)$
  and by Lemma \ref{lemma:lowerboundgraph}, the graph therefore also
  is in $\mathcal{G}(n,k,\log(n/k))$ and it has diameter at most
  $O(\log n)$. We can therefore prove the claim of the theorem by
  providing an appropriate lower bound on $T$. We reduce the problem
  to the two-party set disjointness problem by describing how Alice
  and Bob can together simulate the given distributed protocol.

  Initially, only the nodes $0,\dots,k-1$ depend on the input $X$ of
  Alice and only the nodes $n-k,\dots,n-1$ depend on the input $Y$ of
  Bob. The inputs of all other nodes are known. Initially, Alice
  therefore knows the inputs of all nodes in $\set{0,\dots,n-k-1}$ and
  Bob knows the inputs of all nodes in $\set{k,\dots,n-1}$. Thus, by Theorem
  \ref{thm:simulation}, for $T\leq (n-2k)/(2k)$, there exists a
  $2BcT=O(TB\log n)$-bit public coin two-party protocol
  between Alice and Bob that $k$-solves the problem of finding an
  $\alpha$-approximate minimum cut in $G$. However, since at the end
  of such a protocol, Alice and Bob know the unique minimum cut
  $(S,V\setminus S)$, they can also use it to find the element $z\in
  X\cap Y$. We have seen that this requires them to exchange at least
  $\Omega(k)$ bits and we thus get a lower bound of $T=\Omega(k/(B\log
  n))$ on $T$. Recall that we also need to guarantee that $T\leq
  (n-2k)/(2k)$. We choose $k=\Theta(\sqrt{nB\log n})$ to obtain the
  lower bound claimed by the theorem statement. Note that the lower
  bound even applies if the size $\ell$ of the minimum cut is known.

  If the size of the minimum cut is now known and the task of an
  algorithm is to approximate the size of the minimum cut, we can
  apply exactly the same reduction. This time, we do not use the
  promise that $|X\cap Y|=1$, but only that $|X\cap Y|\leq 1$. The
  size of the minimum cut is $\ell$ if $X$ and $Y$ intersect and it is
  at least $\alpha\ell+1$ if they are disjoint. Approximating the
  minimum cut size therefore is exactly equivalent to solving set
  disjointness in this case.
\end{proof}

\subsection{Lower Bound for Approximating Minimum Cut: Simple Unweighted Graphs}
We next present our lower bound for approximating the minimum cut
problem in unweighted simple graphs.

\begin{theorem}\label{thm:cutlowersimple}
  In unweighted simple graphs, for any $\alpha\geq1$ and $\lambda\geq
  1$, computing an $\alpha$-approximation of $\lambda$ or finding an
  $\alpha$-approximate minimum cut (even if $\lambda$ is known)
  requires at least $\Omega\big(D + \sqrt{n/(B\alpha\lambda\log
    n)}\big)$ rounds, even in networks of diameter
  $D=O(\log n+ \frac{1}{\lambda}\cdot\sqrt{n/(B\alpha\lambda\log n)})$.
\end{theorem}
\begin{proof}[Proof Sketch]
  The basic proof argument is the same as the proof of
  Theorem \ref{thm:cutlowerweighted}. We therefore only describe the
  differences between the proofs. Because in a simple unweighted
  graph, we cannot add edges with different weights and we cannot add
  multiple edges, we have to adapt the construction. Assume that
  $\alpha\geq 1$ and $\lambda\geq 1$ are given. First note that for
  $\lambda=O(1)$, the statement of the theorem is trivial as
  $\Omega(D)$ clearly is a lower bound for approximating the edge
  connectivity or finding an approximate minimum cut. We can therefore
  assume that $\lambda$ is sufficiently large.

  We adapt the construction of the network $G$ to get a simple graph
  $G'$ as follows. First, every node of $G$ is replaced by a clique of
  size $\alpha\lambda+1$. All the ``path'' edges $e\in E_{H,1}$ are
  replaced by complete bipartite graphs between the cliques
  corresponding to the two nodes connected by $e$ in $G$. Let us again
  assume that $k$ is the number of paths and that each of these paths
  is of length $\ell$. Instead of $\ell k$ nodes, the new graph $G'$
  therefore has $\ell k (\alpha\lambda+1)$ nodes. For each edge $e\in
  E_{H,2}$---the edges that are used to reduce the diameter of the
  graph induced by the first path---, we add a single edge between two
  nodes of the corresponding cliques. Adding only one edge suffices to
  reduce the diameter of the graph induced by the cliques of the first
  path. For the edges in $E_{H,3}$, the adaptation is slightly
  larger. The edges among the first $k$ nodes and the last $k$ nodes
  of $G$ that are used to encode the set disjointness instance $(X,Y)$
  into the graph are adapted as follows. Each edge of weight
  $\alpha\ell+1$ is replaced by a complete bipartite subgraph, whereas
  each edge of weight $1$ is replaced by a single edge connecting the
  corresponding cliques in $G'$. For the remaining edges, we introduce
  a parameter $D'=\ell/\lambda$. Instead of vertically connecting each
  of the $\ell$ cliques of all paths to stars (with the center in a
  node on path $0$), we only add some of these vertical
  connections. We already connected the first and the last clique of
  each path. In addition, we add such vertical connections (a single
  edge between the clique on path $0$ and each of the corresponding
  cliques on the other paths) such that: a) between two vertically
  connected ``columns'' there is a distance of at most $2D'$ and b) in
  total, the number of vertically connected ``columns'' is at most
  $\lambda$ (including the first and the last column). Note that
  because the length of the paths is $\ell$, the choice of $D'$
  allows to do so. We now get a graph $G'$ with the following
  properties.
  \begin{itemize}
  \item For each $x\in X\cap Y$, all the vertical connections are
    single edges connecting path $x$ with path $0$. The total number
    of edges connecting the cliques of path $x$ with other nodes is at
    most $\lambda$.
  \item For each $x\not\in X\cap Y$, path $x$ is connected to path $0$
    through a complete bipartite graph
    $K_{\alpha\lambda+1,\alpha\lambda+1}$.
  \item The diameter of $G'$ is $O(D' + \log n)$.
  \end{itemize}
  Let us consider the subgraph $G''$ of $G'$ induced by only the edges
  of all the complete bipartite subgraphs
  $K_{\alpha\lambda+1,\alpha\lambda+1}$ of our construction. If $X\cap
  Y=\emptyset$, $G''$ is connected. Therefore, in this case, the edge
  connectivity of $G''$ (and thus also of $G'$) is at least
  $\alpha\lambda$. If $|X\cap Y|=1$ and if we assume that $z$ is the
  element in $X\cap Y$, $G''$ consists of two components. The first
  component is formed by all the nodes (of the cliques) of path $z$,
  whereas the second component consists of all the remaining
  nodes. By the above observation, the number of edges in $G'$ between
  the two components of $G''$ is at most $\lambda$ and therefore the
  edge connectivity of $G'$ is at most $\lambda$. Also, every other
  edge cut of $G'$ has size at least $\alpha\lambda$. Using the same
  reduction as in Theorem \ref{thm:cutlowerweighted}, we therefore
  obtain the following results
  \begin{itemize}
  \item If the edge connectivity of the network graph is not known,
    approximating it by a factor $\alpha$ requires $\Omega\big(\min\set{k/(B\log n),\ell}\big)$ rounds.
  \item If the edge connectivity $\lambda$ is known, finding a cut of
    size less than $\alpha\lambda$ requires $\Omega\big(\min\set{k/(B\log n),\ell}\big)$ rounds.
  \end{itemize}
  The lower bound then follows by setting $k/(B\log n) =
  \ell$. Together with $n=k\ell(\alpha\lambda+1)$, we get
  $\ell=\Theta(\sqrt{n/(B\alpha\lambda\log n)})$.
\end{proof}


\noindent\textbf{Acknowledgment:}
We thank David Karger for valuable discussions in the early stages of
this work and thank Michael Kapralov for discussing cosmetic
similarities with \cite{RefinementSampling}. We are also grateful to
Boaz Patt-Shamir for pointing out some issues with the lower bounds
and the anonymous reviewers of DISC 2013 for helpful comments on an
earlier version of the manuscript.

\bibliographystyle{abbrv}
\bibliography{references}
\appendix

\section{Missing Parts of \Cref{sec:sampling}}
\label{app:sampling}
\begin{proof}[Proof of \Cref{crl:approx}]
  We run $\Theta(\log^2 n)$ edge-sampling
  experiments: $\Theta(\log n)$ experiments for each sampling
  probability $p_j = 2^{-j}$ where $j \in [1, \Theta(\log
  n)]$. From \Cref{thm:UpperTh}, we know that, if $p_j \geq
  \Omega\big(\frac{\log n}{\C}\big)$, the sampled graph is connected
  with high probability. On the other hand, by focusing on just one
  minimum cut, we see that if $p_j \leq \frac{1}{\C}$, then the probability
  that the sampled graph is connected is at most $3/4$. Let $p^*$ be
  the smallest sampling probability $p_j$ such that at least $9/10$ of
  the sampling experiments with probability $p_j$ lead to sampled
  graph being connected. With high probability, $\tilde{\C} :=
  \frac{1}{p*}$ is an $O(\log n)$-approximation of the
  edge-connectivity. To check whether each sampled graph is connected,
  we use Thurimella's connectivity-tester (refer to
  \Cref{subsec:Thurimella}), and doing that for $\Theta(\log^2 n)$
  different sampled graphs requires $O(D + \sqrt{n} \log^2 n \log^{*}
  n)$ rounds.
\end{proof}

\begin{algorithm}[H]
\caption{An $O(\log n)$ Approximation Algorithm for the Edge-Connectivity}
\begin{algorithmic}[1]
\algorithmsize
	\Statex
	\For{$i=1$ to $\log n$} 
		\For{$j=1$ to $4\log n$}
			\State Choose subset $E^j_i\subseteq E$ by adding each edge $e\in E$ to $E^j_i$ independently with probability $2^{-i}$
		\EndFor
	\EndFor
	\Statex
	\State Run Thurimella's connectivity-tester on graph $G$ with $\Theta(\log^2 n)$ subgraphs $H^j_i=(V, E^j_i)$, in $O(D + \sqrt{n} \log^2 n \log^{*} n)$ rounds.
	\Comment{Refer to \Cref{subsec:Thurimella} for Thurimella's connectivity-tester algorithm.}
	\Statex
	\For{$i=1$ to $\Theta(\log n)$} 
		\For{$j=1$ to $c\log n$}
			\If{graph $G=(V,E_i^j)$ is connected}
				\State $X_i^j\gets 1$
			\Else
				\State $X_i^j\gets 0$
			\EndIf
			\State$ X_i \gets \sum_{j=1}^{c\log n} X_i^j$
		\EndFor
	\EndFor
	\State $i^{*}\gets \arg\max_{i \in [1, \Theta(\log n)]} (X_i \geq 9c\log n/10)$	
	\State $\tilde{\C}\gets 2^{i^{*}}$
	\Statex	
	\State Return $\tilde{\C}$
\label{alg:SizeApprox}
\end{algorithmic}
\end{algorithm}

\section{Missing Parts of \Cref{sec:layeringCut} }\label{App:LayeringCut}
\begin{proof}[Proof of \Cref{lem:CutTester}]
  If a cut $(\mathcal{C}, V\setminus \mathcal{C})$ has size at most
  $\kappa /(1+\delta)$, then the probability that $(\mathcal{C},
  V\setminus \mathcal{C}) \cap E_j \neq \emptyset$ is at most
  $1-(1-p')^{\frac{\kappa}{1+\delta}} = 1-2^{-\, \frac{1}{1+\delta}}
  \leq 0.5 -\frac{\delta}{4}$. On the other hand, if cut
  $(\mathcal{C}, V\setminus \mathcal{C})$ has size at least
  $((1+\delta)\kappa)$, then the probability that $(\mathcal{C},
  V\setminus \mathcal{C}) \cap E_j \neq \emptyset$ is at least
  $1-(1-p')^{(1+\delta)\kappa} \geq 1-2^{-{1+\delta}} \geq
  0.5+\frac{\delta}{4}$.
  This $\Theta(\delta)$ difference between these probabilities gives
  us our basic tool for distinguishing the two cases. Since we repeat
  the experiment presented in \Cref{subsec:CutTester} for
  $\Theta(\frac{\log n}{\delta^2})$ times, an application of
  Hoeffding's inequality shows that if cut $(\mathcal{C}, V\setminus
  \mathcal{C})$ has size at most $\kappa /(1+\delta)$, the test passes
  w.h.p., and if cut $(\mathcal{C}, V\setminus \mathcal{C})$ has size
  at least $\kappa (1+\delta)$, then, w.h.p., the test does not pass.
\end{proof}
\begin{algorithm}[t]
\caption{Distributed cut tester vs.\ threshold $\kappa$ @ node $v$ }
\begin{algorithmic}[1]
\algorithmsize
	\Statex Given a subgraph $G'=(V, E')$ where $E'\subseteq E$, and a threshold $\kappa$
	\Statex
	\State $v.componentID \gets$ the smallest id in the component of $G'$ that contains $v$
	\Comment{Using Thurimella's Component Identification Alg.}
	\Statex
	\For{$j=1$ {\bf to} $c\log(n)/\delta^2$} 
		\State Choose subset $E_i\subseteq E\setminus E'$ by adding each edge $e\in E\setminus E'$ to $E_j$ independently with probability $1-2^{-\frac{1}{\kappa}}$
	\EndFor
	\Statex
	\State $\ell^{max}_j(v) \gets$ the largest $componentID$ in the connected component of $H_i=(V, E' \cup E_i)$ that contains $v$
	\State $\ell^{min}_j(v) \gets$ the smallest $componentID$ in the connected component of $H_i=(V, E' \cup E_i)$ that contains $v$
	\Statex \Comment{Using Thurimella's Component Identification on the $\Theta(\log n)$ values of $i$, simultaneously. (cf.\ \Cref{subsec:Thurimella})}
	\Statex
	\State $X_{i} \gets 0$
	\For{$i=1$ {\bf to} $\alpha\log n$} 
			\If{$\ell^{max}_j(v) \neq v.componentID$ {\bf or} $\ell^{min}_j(v) \neq v.componentID$} $X_i\gets X_i +1$
			\EndIf
	\EndFor
	\Statex
	\State Test passes @ node $v$ iff $X_i \leq \frac{c\log n}{2\delta^2}$
	\label{alg:CutTester}
	\end{algorithmic}
\end{algorithm}

\section{Missing Parts of \Cref{sec:Matula}}\label{App:Matula}
\begin{proof}[Proof of \Cref{lem:ThurimellaCertificate}]
The idea of Thurimella's original sparse certificate-algorithm\cite{Thurimella97} is relatively simple: $E^*$ is made of the edges of $k$ MSTs that are
found in $k$ iterations. Initially, we set $E^*=\emptyset$. In each
iteration, we assign weight $0$ to the edges in $E \setminus E^*$ and
weight $1$ to the edges in $E^*$. In each iteration, we find a new MST
with respect to the new weights using the MST algorithm of
\cite{KuttenPeleg95}, and add the edges of this MST to $E^*$. Because of the weights, each MST tries to avoid using the edges
that are already in $E^*$. In particular, if in one iteration, there
exist two edges $e,e'$, a cut $(\mathcal{C}, V\setminus \mathcal{C})$
such that $e, e'\in (\mathcal{C}, V\setminus \mathcal{C})$ and $e\in
E^*$ but $e' \notin E^*$, then the new MST will not contain $e$ but will contain an edge $e'' \in (E \setminus E^{*}) \cap (\mathcal{C}, V\setminus \mathcal{C})$. This is because, MST will prefer $e''$ to $e$ and there is at least one such $e''$, namely edge $e'$. As a result, if there is a cut with
size at most $k$, in each MST, at least one edge of the cut gets added
to $E^*$, until all edges of the cut are in $E^*$.

To solve our generalized version of sparse certificate, we modify the algorithm in the following
way. As before, we construct the set $E^*$ iteratively such that at
the beginning $E^*=\emptyset$. In each iteration, we give weight $0$
to edges of $E_c$, weight $1$ to edges of $E \setminus (E_c\cup E^*)$
and weight $2$ to edges in $E^*$. Moreover, in each iteration, if the
newly found MST is $T$, we only add edges in $T \setminus E_c$ to the
set $E^*$. Note that if for an edge $e=\set{u,v}\in E$, nodes $u$ and
$v$ correspond to the same node of the edge-contracted graph $G'$, then
edge $e$ will never be added to $E^*$ as either it is in $E_c$ or $u$
and $v$ are connected via a path made of edges in $E_c$ and thus, in
each MST, that path is always preferred to $e$. Moreover, if there is
a cut $(\mathcal{C}, V\setminus \mathcal{C})$ of $G$ such that
$(\mathcal{C}, V\setminus \mathcal{C}) \cap E_c = \emptyset$ and there
are two edges $e, e'\in (\mathcal{C}, V\setminus \mathcal{C})$ such
that $e\in E^*$ but $e' \notin E^*$, then the new MST will not contain
$e$ but will contain an edge $e'' \in (E \setminus E^{*}) \cap (\mathcal{C}, V\setminus \mathcal{C})$.
\end{proof}

\section{Information Dissemination Lower Bound}
\label{sec:disseminationlower}

Here, we use Theorem \ref{thm:simulation} to show a lower bound on a
basic information dissemination task in $\lambda$-edge connected
networks. We show that even if such networks have a small diameter, in
general, for $s$ sufficiently large, disseminating $s$ bits requires
time at least $\Omega(n/\lambda)$. As a corollary of our lower bound,
we also obtain a lower bound on the diameter of the graph induced
after independently sampling each edge of a $\lambda$-edge connected
graph with some probability.

 The following theorem follows from Theorem
\ref{thm:simulation} by taking a $\lambda$-connected network $H\in
\mathcal{H}(n,k)$ (choosing $k$ as small as possible) and by adding
some edges to create a network of small diameter.

\begin{theorem}\label{thm:dissemination}
  For any $\lambda\geq 1$, there exist weighted simple $n$-node graphs
  (or equivalently unweighted multigraphs) $G=(V,E)$ with edge
  connectivity at least $\lambda$ and diameter $D=O(\log n)$ such that
  for two distinguished nodes $s,t\in V$, sending $K$ bits of
  information from $s$ to $t$ requires time at least
  $\Omega\left(\min\left\{\frac{K}{B\log n},n\right\}\right)$. For
  unweighted simple graphs, the same problem requires
  $\Omega\left(\min\left\{\frac{K}{B\log
        n},\frac{n}{\lambda}\right\}\right)$ rounds.
\end{theorem}

\begin{proof}
  We consider the weighted simple graph $H=(V,E_H,w_H)$ with node set
  $V=\set{0,\dots,n-1}$, edge set $E_H=\set{i,j\in V: |i-j|= 1}$, and
  $w_H(e)=\lambda$ for all $e\in E_H$, as well as the unweighted
  simple graph $H'=(V,E_H')$ with node set $V=\set{1,\dots,n}$ and
  edge set $E_H'=\set{i,j\in V: |i-j|\leq \lambda}$. Graphs $H$ and
  $H'$ both have edge connectivity $\lambda$, but they still both have
  large diameter. In order to get the diameter to $O(log n)$, we
  proceed as follows in both cases. For every two nodes
  $i,j\in\set{0,\dots,n-1}$ for which there exists an integer
  $\ell\geq 1$ such that $i\equiv 0\!\pmod{2^\ell}$ and $j=i+2^\ell$,
  if $i$ and $j$ are not connected by an edge in $H$ (or in $H'$), we
  add an edge of weight $1$ between $i$ and $j$. The resulting graphs
  have diameter $O(\log n)$ and we have $H\in\mathcal{G}(n,1,O(\log
  n))$ and $H'\in\mathcal{G}(n,\lambda,O(\log n))$.

  In both cases, we choose the two distinguished nodes $s$ and $t$ as
  $s=0$ and $t=n-1$. Sending $K$ bits from $s$ to $t$ can then be
  modelled as the following distributed task in $G(H)$. Initially node
  $0$ gets an arbitrary $K$-bit string as input. The input of every
  other node of $G(H)$ is the empty string. To solve the task, the
  output of node $n-1$ has to be equal to the input of node $0$ and
  all other nodes need to output the empty string.

  Assume that there is a (potentially randomized) distributed protocol
  that solves the given information dissemination task in $T$ rounds
  with error probability $\eps$.  For $T\leq (n-2)/2$ (in the case of
  $H$) and $T\leq (n-2)/(2\lambda)$ (in the case of $H'$), Theorem
  \ref{thm:simulation} therefore implies that there exists a
  public-coin two-party protocol that $1$-solves the given task with
  communication complexity at most $O(BT\log n)$. In such a protocol,
  only Alice gets the input of node $0$ and Bob has to compute the
  output of node $n-1$. Consequently Alice needs to send $K$ bits to
  Bob and we therefore need to have $BT\log n = \Omega(K)$. Together
  with the upper bounds on $T$ required to apply Theorem
  \ref{thm:simulation}, the claim of the theorem then follows.
\end{proof}

From Theorem \ref{thm:dissemination}, we also get an
upper bound on the diameter of when partitioning a graph into
edge-disjoint subgraphs.

\begin{corollary}\label{cor:partitionlower}
  There are unweighted $\lambda$-edge connected simple graphs $G$ and
  unweighted $\lambda$-edge connected multigraphs $G'$ such that when
  partitioning the edges of $G$ or $G'$ into $\ell\geq\gamma\log n$
  spanning subgraphs, for a sufficiently large constant $\gamma$, at
  least one of the subgraphs has diameter $\Omega(n/\lambda)$ in the
  case of $G$, and $\Omega(n)$ in the case of $G'$.
\end{corollary}

\begin{proof}
  Assume that we are given a graph $G=(V,E)$ and a partition of the
  edges $E$ into $\ell$ spanning subgraphs such that each subgraph has
  diameter at most $D$. Consider two nodes $s$ and $t$ of $G$. We can
  use the partition of $E$ to design a protocol for sending $K$ bits
  from $s$ to $t$ in $O(D+K/\ell)$ rounds as follows. The $K$ bits are
  divided into equal parts of size $K/\ell$ bits. Each part is in
  parallel sent on one of the $\ell$ parts. Using pipelining this can
  be done in time $T=O(D+K/\ell)$. From Theorem
  \ref{thm:dissemination}, we know that for simple graphs,
  $T=\Omega\left(\min\left\{\frac{K}{B\log
        n},\frac{n}{\lambda}\right\}\right)$. Choosing
  $K=Bn\log(n)/\lambda$ and $\gamma$ sufficiently small then implies
  the claimed lower bound on $D$. The argument for multigraphs is done
  in the same way by using the stronger respective lower bound in
  Theorem \ref{thm:dissemination}
\end{proof}

\paragraph{Remark:} When partitioning the edges of a graph $G$ in a
random way such that each edge is independently assigned to a
uniformly chosen subgraph, each subgraph corresponds to the induced
graph that is obtained if each edge of $G$ is sampled with probability
$p=1/\ell$. As a consequence, the diameter lower bounds of Corollary
\ref{cor:partitionlower} also holds with at least constant probability
when considering the graph obtained when sampling each edge of a
$\lambda$-connected graph with probability $p=1/\ell\leq 1/(\gamma\log
n)$.


\end{document}